\documentclass[12pt]{article}
\usepackage{amsfonts}
\usepackage{amssymb}
\usepackage{mathrsfs}
\usepackage{srcltx}
\usepackage{amsmath}
\usepackage{amsthm}
\usepackage{amstext}
\usepackage{amsopn}
\usepackage{graphicx}

\newtheorem{theorem}{Theorem}[section]
\newtheorem{lemma}[theorem]{Lemma}

\newtheorem{corollary}[theorem]{Corollary}

\theoremstyle{definition}

\theoremstyle{remark}

\numberwithin{equation}{section}

\newcommand{\ba}{\begin{array}}
\newcommand{\ea}{\end{array}}

\begin{document}

\date{}
\title{ \bf\large{ Periodic Solutions for $N$-Body-Type Problems}}

\author{Fengying Li\textsuperscript{1}\footnote{Email:lify0308@163.com }\ \ Shiqing Zhang\footnote{Email:zhangshiqing@msn.com}\textsuperscript{2}
 \\
{\small \textsuperscript{1}School of Economic and Mathematics,Southwestern University of Finance and Economics,\hfill{\ }}\\
\ \ {\small Chengdu, Sichuan, 611130, P.R.China.\hfill{\ }}\\
{\small \textsuperscript{2} Department of Mathematics, Sichuan University,\hfill{\ }}\\
\ \ {\small Chengdu, Sichuan, 610064, P.R.China.\hfill {\ }}}

\maketitle
\begin{abstract}
{We consider non-autonomous $N$-body-type problems with strong force type
potentials at the origin and sub-quadratic growth at infinity, and using Ljusternik-Schnirelmann theory, we prove the existence of unbounded sequences of critical values for the Lagrangian action corresponding to non-collision periodic solutions.}

 \noindent{\emph{Keywords}}: Periodic solutions, N-body type problems , variational methods.

 \noindent{\emph{2000 Mathematical Subject Classification}}: 34C15, 34C25, 58F,70F07.
\end{abstract}

\section {Introduction and Main Result}

In the 1975 paper of Gordon (\cite{13}), we find the first prominent use of variational methods in the study of
periodic solutions of Newtonian equations with singular potential
 $V(t,x)\in C^1([0,T]\times
(R^n\backslash S),R)$ :
\begin{equation}\label{1.1}
\left\{\begin{array}{l}
\ddot{x}+V^{\prime}(t,x)=0,\ \ \ \ x\in R^n\\
x(t+T)=x(t),
\end{array}\right.
\end{equation}
where $V(t+T,x)=V(t,x)$ satisfies the following Gordon's strong-force (SF)
condition:

There exists a neighborhood $N$ of the set $S$ and a $C^2$ function $U$ on
$N-S$ such that
\begin{enumerate}
\item[(i)] $U(x)\rightarrow -\infty$ as $x\rightarrow S$,

\item[(ii)] $-V(t,x)\geq |\nabla U(x)|^2, \forall x\in N\backslash S$.
\end{enumerate}
Using variational minimizing methods,Gordon proved the following theorem:

\begin{theorem}\label{th1.1}
{\bf(Gordon)}\ \ Under the above conditions and $(V_1)$:\\
$$ V(t,x)<0,x\neq 0,$$
 we have that there exist periodic solutions which
tie (wind around) $S$ and have arbitrary given topological
(homotopy) type and  given period.
\end{theorem}

Ambrosetti-CotiZelati (\cite{1},\cite{2}) used Morse theory to generalize
Gordon's result and obtained the following theorem:

\begin{theorem}\label{th1.2}
 Assume $V\in C^2([0,T]\times
R^n,R),V(t+T,x)=V(t,x)$ satisfies Gordon's strong force condition
and $(V_i)$:
$$|V(t,x)|,\ \ \ \ |V_x(t,x)|\rightarrow 0$$
uniformly for all $t$ as $\|x\|\rightarrow\infty$, and $\exists R_1>0$
s.t.
$$V(t,x)<0,\ \ \ \ \forall \|x\|\geq R_1,$$
then (\ref{1.1}) has infinitely many $T$-periodic solutions.
\end{theorem}

Motivated by Gordon (\cite{13}) and Ambrosetti-Coti Zelati (\cite{1},\cite{2}),
Jiang M. Y. (\cite{15}) continues the application of Morse theory to arrive at his theorem
which improves the condition $(V_i)$:

\begin{theorem}\label{th1.3}
 Let $\Omega$ be an open subset in $R^n$ with
compact complement $C=R^n\backslash\Omega$, $n\geq 2$. Assume $V\in
C^2([0,2\pi ]\times\Omega ,R)$, $V(t+2\pi ,x)=V(t,x)$, and
\begin{enumerate}
\item[$(A_1)$.] $\exists R_0$ such that
$$\rm{sup}\{|V(t,x)|+|V^{\prime}_x(t,x)|\ |(t,x)\in [0,2\pi
]\times (R^n\backslash B_{R_0})\}<+\infty$$
\item[$(A_2)$.] $V$ satisfies Gordon's strong force condition.
\end{enumerate}
Then (\ref{1.1}) has infinitely many $2\pi$-periodic solutions.
\end{theorem}

 Using Ljusternik-Schnirelman theory, Majer\cite{16} got the following
 result which improves the condition $(A_1)$:
\begin{theorem}\label{th1.4}
Assume  $W\in C^1([0,T ]\times (R^N
\backslash 0),R)$ satisfies
\begin{enumerate}
\item[(i).] $W(t+T,x)=W(t,x).$
\item[(ii).] $\exists c\in R,\theta<2,r>0$, such that
$$W(t,x)\leq c|x|^{\theta},\ \ W^{\prime}(t,x)x-2W(t,x)\leq
c|x|^{\theta}, \ \ \forall |x|>r, \forall t>0. $$
\item[(iii).] If $a<(\frac{\pi}{T})^2$
\end{enumerate}
Then
$$ \ddot u+au+W^{\prime}(t,u)=0$$
has infinitely many $T$-periodic solutions.
\end{theorem}

For 3-body type problem, Bahri-Rabinowitz (\cite{3}) used Morse theory
to prove:

\begin{theorem}\label{th1.5}
{\bf(Bahri-Rabinowitz)} Let
$V(q)=\frac{1}{2}\sum\limits_{1\leq i\neq j\leq
3}V_{ij}(q_i-q_j)$. Assume $V_{ij}$ satisfies
\begin{enumerate}
\item[($V_1$).] $ V_{ij}\in C^2(R^l\backslash\{0\},R).$
\item[($V_2$).] $ V_{ij}<0.$
\item[($V_3$).] $V_{ij}(q)$, $V^{\prime}_{ij}(q)\rightarrow 0$ as
$|q|\rightarrow\infty.$
\item[($V_4$).] $ V_{ij}(q)\rightarrow -\infty$ as $q\rightarrow 0.$
\item[($V_5$).] For $\forall M>0$, $\exists R>0$, s.t.
$$V^{\prime}_{ij}(q)\cdot q>M|V_{ij}(q)|,\ \ |q|>R.$$
\item[($V_6$).] $\exists U_{ij}\in C^1(R^l\backslash\{0\},R)$, s.t.
$$U_{ij}(q)\rightarrow\infty\ \rm{as} \ \ q\rightarrow 0,\ \rm{and}\
-V_{ij}\geq |U^{\prime}_{ij}|^2.$$
\end{enumerate}
 Then for any given $T>0$,
\begin{equation}\label{1.2}
\ddot{q}_i+\frac{\partial V(q)}{\partial q_i}=0
\end{equation}
has infinitely many T-periodic noncollision solutions.
\end{theorem}

We say that a function $X(t)=(x_1(t),\cdots,x_N(t))\in
C^2(R,(R^k)^N)$ is a  non-collision $T$-periodic solution of (\ref{1.2})
if $X(t)$ satisfies $x_i(t)\neq
x_j(t)$ for all $i\neq j$ and $t\in R$ ,and satisfies the equation $(1.2)$ and is indeed $T$ periodic.

Majer-Terracini (\cite{17}) generalized the result of Bahri-Rabinowitz
to $n$-body type problems:
\begin{equation}\label{1.3}
\ddot{x}_i(t)+\nabla_{x_i}V(t,x_1(t),\cdots,x_N(t))=0,\ \ \ \
x_i(t)\in R^k,\ \ \ \ i=1,\cdots,N.
\end{equation}

They proved the following theorem:

\begin{theorem}\label{th1.6}
 Assume $k\geq 3$, and $V_{ij}\in C^1((R^k\backslash 0)\times
R,R)$ are $T$-periodic in $t$, $V$ satisfies
\begin{enumerate}
\item[($V_1$).] $ V_{ij}(t,x)=V_{ji}(t,-x), \forall x\in R^k\backslash\{0\}.$
\item[($V_2$).] $ V_{ij}(t,x)\leq 0, \forall x\in R^k\backslash\{0\}.$
\item[($V_3$).] $ V_{ij}(t,\xi )\rightarrow -\infty$ uniformly in t as $|\xi
|\rightarrow 0$, for all $1\leq i\neq j\leq N,$ and
$V_{ij}$ satisfies Gordon's strong force condition.
\item[($V_4$).] $\exists\rho >0$,
$\exists\theta\in\left[0,\frac{\pi}{2}\right)$ s.t. any $(\nabla
V_{ij}(t,x),x)\leq\theta , \forall x,|x|>\rho$.
\end{enumerate}
Then (1.3) has at least one $T$-periodic non-collision solution.
\end{theorem}

For symmetrical potentials,Fadell-Husseini[11],Zhang-Zhou[25] proved that
\begin{theorem}\label{th1.7}
 We assume $V_{ij}$ satisfies the following conditions:
\begin{enumerate}
\item[(V1).] $V(t,x)=\frac{1}{2}\sum\limits_{1\leq
i\neq j\leq N}V_{ij}(t,x_i-x_j)$.

\item[(V2).] $V_{ij}\in C^1(R\times (R^k-\{0\});R)$, for all $1\leq i\neq
j\leq N$.

\item[(V3).] $V_{ij}(t,\xi )\rightarrow -\infty$ uniformly in t as $|\xi
|\rightarrow 0$, for all $1\leq i\neq j\leq N.$

\item[(V4).]  $V_{ij}(t,\xi )\leq 0, 1\leq i\neq j\leq N,\xi\neq 0$.

\item[(V5).] the strong force condition (see [13]) holds for $V_{ij}$.

\item[(V6).] $V_{ij}(t+T/2,-\xi)=V_{ij}(t,\xi)$.
\end{enumerate}
Then there exist unbounded sequences of critical values for the Lagrangian action
corresponding to non-collision periodic solutions for (\ref{1.3}).
\end{theorem}

In this paper, we consider a relaxation of condition $(V4)$ which required the potentials
to be non-positive, but still maintain the potentials have some growth so that the result in Theorem \ref{th1.7}
still holds. We use Majer's abstract critical point theorem to study the N-body-type problem. The key difficulty is in proving the local Palais-Smale condition, but we are able to secure the following:

 \begin{theorem}\label{th1.8}
Assume $V_{ij}$ satisfies $(V1)-(V3)$ ,$(V5),(V6)$ and

$(V4)'$\ $\exists g>0$, $\theta<2$ ,$r>0$, s.t.

$$V_{ij}(t,\xi)\leq gm_im_j|\xi|^{\theta} ,|\xi |>r.$$
Then there exist unbounded sequences of critical values for the Lagrangian action
corresponding to non-collision periodic solutions for (\ref{1.3}).
\end{theorem}

Notice that the condition $(V4)'$ in our theorem 1.8 is a kind of growth condition which weak the ordinary condition on potentials which need non-positive,we have the following Corollary:

\begin{corollary}\label{c1.1}
 Let $\alpha\geq 2;r_1>0,r_2>r_1;a,g>0,\theta<2$ and
$$V(t,x)=\frac{1}{2}\sum\limits_{1\leq
i\neq j\leq N}V_{ij}(x_i-x_j).$$
$$V_{ij}(\xi)\in C^1(R^k-\{0\},R)$$
satisfies that
$$V_{ij}(\xi)=-am_im_j|\xi|^{-\alpha},|\xi|<r_1;$$
$$V_{ij}(\xi)=gm_im_j|\xi|^{\theta},|\xi|\geq r_2>r_1.$$
Then the assumptions and the result of Theorem \ref{th1.8} holds.
\end{corollary}

 We also notice that for Newtonian type potentials ,there were a deep and lively literature; for example, \cite{4}, \cite{6}-\cite{10}, \cite{12}-\cite{13}, \cite{18}-\cite{24}.
\section{Some Lemmas}

We introduce spaces

$$E=\{(x_1,\cdots,x_N)|x_i\in H^1(R/TZ;R^k),x_i(t+T/2)=-x_i(t)\},$$
$$\Delta =\{(x_1,\cdots,x_N)|x_i\in H^1(R/TZ;R^k), x_i(t)\neq x_j(t),\ \ \forall t,i\neq j\},$$
where $H^1(R/TZ;R^k)$ is the metric completion of smooth
$T$-periodic functions for the norm
$\|x\|_{H^1}=\left(\int^T_0|x(t)|^2+|\dot{x}(t)|^2dt\right)^{1/2}$,
and the functional $f:\Delta\rightarrow R$ is defined by
$$f(x_1,\cdots,x_N)=\sum\limits^N_{i=1}\frac{m_i}{2}\int^T_0|\dot{x}_i(t)|^2dt-\int^T_0V(t,x_1(t),\cdots,x_N(t))dt.$$

Clearly, $E$ is a closed subspace of $H^1(R/TZ; (R^k)^N)$, and so a Hilbert space, while $\Delta$ is an open subset of $E$.

Using a standard argument (for instance, see \cite{25}), it is easy to
prove the following Lemma 2.1:

\begin{lemma}\label{l2.1}
Suppose $(V1)-(V2)$ and $(V6)$ hold, then a
critical point of $f$ in $\Delta$ is a non-collision solution of
(\ref{1.3}).
\end{lemma}

The closed subset $\Gamma =E-\Delta$ of $E$ will be called the
collision set, and a standard argument can be applied to show that
the strong force assumption $(V3)$ implies that
$f(X)\rightarrow +\infty$ when $X$ approaches the collision set
$\Gamma$. More precisely, we have the following lemma.

\begin{lemma}\label{l2.2}
(\cite{13}\cite{25}) Assume V satisfies (V1)-(V3) and (V5). Let $\{X^n\}$ be a sequence in $\Delta$
and $X^n\rightarrow X\in\Gamma$ in both the $C^0$ topology and weak
topology of $E$, then $f(X^n)\rightarrow +\infty$.
\end{lemma}

\begin{lemma}\label{l2.3}
 Assume V satisfies $(V1)-(V3),(V4)^{'},(V5)$ and $(V6)$, then there is a constant $\lambda_0$ depending
on $g,m_i,r,\theta$, s.t. $f$ satisfies the $(PS)_c$ condition for
$c\geq\lambda_0$; that is, any sequence $\{x^k\}\subset\Delta$ satisfying
$f(x^k)\rightarrow c$ and $f^{\prime}(x^k)\rightarrow 0$ is
pre-compact in $H^1$.
\end{lemma}

\begin{proof}
We notice that the arguments in Jiang \cite{15} (also Chang K.C.\cite{5}) and Majer \cite{16} cannot
be directly generalized to the N-body case because of the translation invariance for positions
in N-body problems. Here we must consider the differences and must use different arguments.
By Holder's inequality, we have
$$\sum_{i<j}m_im_j|x_i-x_j|^{\theta}\leq(\sum_{i<j}m_im_j)^{(2-\theta)/2}
(\sum_{i<j}m_im_j|x_i-x_j|^2)^{\frac{\theta}{2}}$$
Calculating,
\begin{eqnarray*}
\sum_{i<j}m_im_j|x_i-x_j|^2&=&\frac{1}{2}\sum_{1\leq i,j\leq N}m_im_j|x_i-x_j|^2\\
&=&\sum_{i=1}^Nm_i\sum_{i=1}^Nm_i|x_i|^2-(\sum_{i=1}^Nm_ix_i)^2\\
&\leq&\sum_{i=1}^Nm_i\sum_{i=1}^Nm_i|x_i|^2
\end{eqnarray*}
and
$$f(x^k)=f(x_1^k,\cdots,x_N^k)=\sum\limits^N_{i=1}\frac{m_i}{2}\int^T_0|\dot{x}_i^k(t)|^2dt-
\int^T_0V(t,x_1^k(t),\cdots,x_N^k(t))dt.$$
Let
 $$\xi_{ij}^k(t)=x_i^k(t)-x_j^k(t)$$
We consider the three possibilities:

 (i). For all $1\leq i,j\leq N$ and for all $t\in [0,T]$, $|\xi_{ij}^k(t)|>r$ when $k$ is large, then by $(V4)'$ and the above inequality,
 we have
 \begin{eqnarray*}
&\sum\limits^N_{i=1}&\frac{m_i}{2}\int^T_0|\dot{x}_i^k(t)|^2dt-
\int^T_0V(t,x_1^k(t),\cdots,x_N^k(t))dt \\
&\geq& \sum\limits^N_{i=1}\frac{m_i}{2}\int^T_0|\dot{x}_i^k(t)|^2dt
-g(\sum_{i<j}m_im_j)^{(2-\theta)/2}[\sum_{i=1}^Nm_i]^{\frac{\theta}{2}}
\int_0^T[\sum_{i=1}^Nm_i|x_i^k|^2]^{\frac{\theta}{2}}dt.
\end{eqnarray*}
 Since $x^k(t+T/2)=-x^k(t)$ implies $\int_0^Tx^k(t)dt=0$, by the Wirtinger's inequality and $f(x^k)\rightarrow c\leq d$, we get
 \begin{equation*}
d\geq \sum\limits^N_{i=1}\frac{m_i}{2}\int^T_0|\dot{x}_i^k(t)|^2dt
-g[\sum_{i<j}m_im_j]^{(2-\theta)/2}[\sum_{i=1}^Nm_i]^{\frac{\theta}{2}}[\frac{T}{2\pi}]^{\theta}
\int_0^T[\sum_{i=1}^Nm_i|\dot{x}_i^k|^2]^{\frac{\theta}{2}}dt.
\end{equation*}
By the assumption $(V4)'$, we know that $\theta<2$, hence we have $e>0$ such that
$$\sum\limits^N_{i=1}\frac{m_i}{2}\int^T_0|\dot{x}_i^k(t)|^2dt\leq e.$$

(ii). There are $1\leq i_0,j_0\leq N$ such that for all $t\in [0,T]$, there holds $|\xi_{i_0j_0}^k(t)|\leq r$ when $k$ is large, then by Lemma \ref{l2.2} and $(V2)$, we have $a>-\infty$ and $0<b<+\infty$ such that for all $t\in [0,T]$,
$$a\leq V_{i_0j_0}(t,\xi_{i_0j_0}^k(t))\leq cm_{i_0}m_{j_0}|\xi_{i_0j_0}^k(t)|^{\theta}\leq b$$
Then for the rest index pairs $(i,j)$ and the corresponding potentials, we can use the above arguments of $(i)$ and notice that we can add some negative terms to estimate the lower bound for the sum of all the potentials satisfying $|\xi_{ij}^k(t)|>r$ :

$$-g(\sum_{i<j}m_im_j)^{(2-\theta)/2}[\sum_{i=1}^Nm_i]^{\frac{\theta}{2}}
[\sum_{i=1}^Nm_i|x_i^k|^2]^{\frac{\theta}{2}}.$$
Now we can consider all cases for the index pairs. We have
$$V(t,x_1^k(t),\cdots,x_N^k(t))\geq \frac{N^2-N}{2}(-b)
-g(\sum_{i<j}m_im_j)^{(2-\theta)/2}[\sum_{i=1}^Nm_i]^{\frac{\theta}{2}}
[\sum_{i=1}^Nm_i|x_i^k|^2]^{\frac{\theta}{2}}.$$
Then taking the integral and using a similar argument as in (i), we can also get $e_1>0$ such that

$$\sum\limits^N_{i=1}\frac{m_i}{2}\int^T_0|\dot{x}_i^k(t)|^2dt\leq e_1.$$

(iii). There are $1\leq i_0,j_0\leq N$ ,$t_1\in [0,T]$ and $t_2\in [0,T]$ such that $|\xi_{i_0j_0}^k(t_1)|>r,$
 $|\xi_{i_0j_0}^k(t_2)|\leq r$ when $k$ is large.\\
 Then
 $$-V(t_1,x_1^k(t_1),\cdots,x_N^k(t_1))\geq -g(\sum_{i<j}m_im_j)^{(2-\theta)/2}[\sum_{i=1}^Nm_i]^{\frac{\theta}{2}}
[\sum_{i=1}^Nm_i|x_i^k(t_1)|^2]^{\frac{\theta}{2}},$$
$$-V(t_2,x_1^k(t_2),\cdots,x_N^k(t_2))\geq -b.$$
Hence for all $t\in [0,T]$, we have

 $$-V(t,x_1^k(t),\cdots,x_N^k(t))\geq -b-g(\sum_{i<j}m_im_j)^{(2-\theta)/2}[\sum_{i=1}^Nm_i]^{\frac{\theta}{2}}
[\sum_{i=1}^Nm_i|x_i^k(t)|^2]^{\frac{\theta}{2}}.$$

Again, after taking the integral, we can find $e_2>0$ such that

$$\sum\limits^N_{i=1}\frac{m_i}{2}\int^T_0|\dot{x}_i^k(t)|^2dt\leq e_2.$$

In all cases, we get the bounded property for
$$\sum\limits^N_{i=1}\frac{m_i}{2}\int^T_0|\dot{x}_i^k(t)|^2dt.$$
This implies $\left\{x^k\right\}$ has a weakly convergent subsequence. To prove the trongly convergent property is more or less standard.
\end{proof}

The following is an abstract critical point theorem which we will
use in the proof of our main result. A proof of this theorem
can be found in Majer\cite{16}.

\begin{lemma}\label{l2.4}
Let $\Delta$ be an
open subset in a Banach space and let $Cat(\Delta)$ denote the category of $\Delta$. Suppose $f$ is a functional on $\Delta$. Assume that
\begin{enumerate}
\item[1.] Cat$\Delta =+\infty$,

\item[2.] For any sequence $\{q_n\}\subset\Delta$ and $q_n\rightarrow
q\in\partial\Delta$, we will have $f(q_n)\rightarrow +\infty$,

\item[3.] For any $K\in R, Cat_{\Delta}(\{q\in\Delta |f(q)\leq
K\})<+\infty$, and

\item[4.] There exists a $\lambda_0\in R$ such that the Palais-Smale
condition holds on the set $\{q\in\Delta |f(q)\geq\lambda_0\}$.
\end{enumerate}
Then $f$ possesses an unbounded sequence of critical values.
\end{lemma}

Fadell-Husseini \cite{11} and Zhang-Zhou \cite{25} proved:

\begin{lemma}\label{l2.5}
If $\Delta$ refers to the open subset defined in our proof of Theorem 1.8, then
$$Cat(\Delta )=+\infty.$$
\end{lemma}

We notice that we can use the similar methods as in Lemma \ref{l2.3} to prove

\begin{lemma}\label{l2.6}
For any $K\in R$ such that $f(q)\leq
K$, there is $ A\geq 0$ such that
$$\sum\limits^N_{i=1}\frac{m_i}{2}\int^T_0|\dot{x}_i(t)|^2dt\leq A.$$
\end{lemma}

 Zhang-Zhou \cite{25} gave:

\begin{lemma}\label{l2.7}
 For any constant
$K\geq 0$, the set $D_K=\{X\in\Delta |\ \|\dot{X}\|_{L^2}\leq K\}$
is of finite category in $\Delta$, i.e.,
$Cat_{\Delta}(D_k)<+\infty$.
\end{lemma}

By the monotone property of category and using Lemma \ref{l2.6} and Lemma \ref{l2.7}, we have

\begin{lemma}\label{l2.8}
 For any $K\in R, Cat_{\Delta}(\{q\in\Delta |f(q)\leq
K\})<+\infty$.
\end{lemma}

The proof of Theorem \ref{th1.8} now follows by Lemmas \ref{l2.1}-\ref{l2.5} and Lemma \ref{l2.8}.

\end{document}